\documentclass[aps,prx,reprint,superscriptaddress,nofootinbib]{revtex4-2}

\makeatletter
\let\auto@bib@innerbib\@empty
\makeatother

\usepackage[T1]{fontenc}
\usepackage[utf8]{inputenc}
\usepackage{amsmath,amssymb,amsthm,mathtools,bm}
\usepackage{booktabs}
\usepackage{graphicx}
\usepackage{placeins}
\usepackage{microtype}
\usepackage{tikz}
\usepackage{hyperref}
\hypersetup{hidelinks}
\usetikzlibrary{positioning}
\newtheorem{theorem}{Theorem}
\newtheorem{proposition}[theorem]{Proposition}
\newtheorem{lemma}[theorem]{Lemma}
\newtheorem{corollary}[theorem]{Corollary}
\theoremstyle{definition}

\newcommand{\Tr}{\operatorname{Tr}}
\newcommand{\supp}{\operatorname{supp}}
\newcommand{\C}{\mathbb{C}}
\newcommand{\F}{\mathbb{F}}
\newcommand{\I}{\mathbb{I}}
\newcommand{\ket}[1]{\lvert #1\rangle}
\newcommand{\bra}[1]{\langle #1\rvert}
\newcommand{\braket}[2]{\langle #1\mid #2\rangle}
\newcommand{\calG}{\mathcal{G}}
\newcommand{\calP}{\mathcal{P}}
\newcommand{\calX}{\mathcal{X}}
\newcommand{\calY}{\mathcal{Y}}
\newcommand{\calA}{\mathcal{A}}
\newcommand{\calB}{\mathcal{B}}
\newcommand{\ind}[1]{\mathbf{1}_{\{#1\}}}

\begin{document}

\title{Perfect Games in Dimension-Bounded Communication}

\author{Emmanuel Zambrini Cruzeiro}
\email{emmanuel.cruzeiro@lx.it.pt}
\affiliation{Quantum Physics of Information (QPI) Group, Instituto de Telecomunica\c{c}\~oes, Lisbon, Portugal}
\affiliation{Quantum Information and Quantum Optics (QIQO) Laboratory, Instituto Superior T\'ecnico, Lisbon, Portugal}

\begin{abstract}
Perfect prepare-and-measure games exhibit an all-or-nothing quantum advantage: a quantum system of dimension $d$ satisfies every prescribed winning constraint, whereas a classical $d$-level message cannot. We establish three structural results for such forbidden-output support constraints. First, every binary-output support game reduces exactly to a conflict graph: perfect classical realization with a $d$-level message is equivalent to $d$-colorability, perfect $d$-dimensional quantum realization is equivalent to a $d$-dimensional orthogonal representation, and the minimum number of Bob inputs realizing a fixed conflict graph is its edge biclique-cover number. Second, for an arbitrary finite output alphabet, every perfect qubit strategy admits a perfect classical-bit realization. Third, with at most two Bob inputs every perfect qutrit strategy admits a perfect classical-trit realization. An explicit seven-preparation qutrit game with three Bob inputs satisfies $C_3=20<Q_3=S=21$, so three Bob inputs are necessary and sufficient in dimension three. Conversely, a two-input, six-output game in dimension five satisfies $C_5=39<Q_5=S=40$. Thus the smallest dimension admitting a two-input perfect same-dimensional separation is either four or five; the four-dimensional case remains open. As a flagship binary application, the $13$-ray qutrit graph yields a compressed game $(X,Y,B)=(13,8,2)$ with $C_3=39<Q_3=S=40$, and eight Bob inputs are minimal among all binary-output realizations of that graph. Graph extensions demonstrate the mechanism in every dimension, while Torpedo and antidistinguishability games illustrate the genuinely nonbinary regime. These results connect exact communication, graph coloring, contextuality, state exclusion, and zero-error information theory.
\end{abstract}

\maketitle

\section{Introduction}

Dimension-bounded prepare-and-measure (PM) experiments are among the simplest operational settings in which classical and quantum communication can be compared. Alice receives an input $x$, encodes it into a physical system of bounded dimension, and sends the system to Bob. Bob receives an input $y$, performs a measurement, and outputs $b$. The observed data are conditional probabilities $p(b|x,y)$. This framework underlies dimension witnesses, semi-device-independent protocols, and tests of restricted information capacity \cite{Gallego2010}.

Most PM quantum advantages are quantitative: a quantum $d$-level system achieves a larger score than a classical $d$-level message, but neither reaches the algebraic maximum. Here we study the extremal case. A \emph{perfect game} is a finite PM task in which a $d$-dimensional quantum strategy satisfies every prescribed winning constraint, while no classical $d$-level message does. Such a separation depends only on specified forbidden-output constraints; a perfect realization may have additional accidental zeros among its winning outputs. These problems are naturally connected with zero-error communication and graph coloring \cite{Stahlke2016}.

This support-constraint problem is distinct from exact simulation of complete probability distributions. Arbitrary qubit PM statistics can require four classical message values, even with shared randomness \cite{Renner2023}, and recent work has reduced the number of preparations and measurements needed to demonstrate this full-statistics cost and obtained stronger qutrit lower bounds \cite{Schlosser2026}. We instead ask how much classical communication is required to reproduce only the forbidden-output constraints needed for perfect winning. The distinction is sharp: we prove that every finite perfect qubit protocol, with any finite output alphabet, has a deterministic classical-bit realization.

Graph-defined promise-equality problems are an important known special case. Their one-round classical and quantum message dimensions are characterized by chromatic number and complex orthogonal rank, respectively \cite{deWolf2001,Briet2015}. Very recently, Prakash introduced a related equal-versus-adjacent promise problem for quantum finite automata, obtaining classical memory $\chi(G)$ and a quantum construction of dimension at most $\xi_{\C}(G)+1$ in a formal-language-recognition setting \cite{Prakash2026}. Our graph theorem is not a new derivation of that graph-defined special case. Its contribution is the converse structural statement that \emph{every} binary-output PM support game, even when no graph is specified a priori, reduces exactly to a conflict graph. Perfect classical realization with a $d$-level message is equivalent to $d$-colorability, whereas perfect $d$-dimensional quantum realization is equivalent to a $d$-dimensional orthogonal representation. Moreover, we prove that the minimum number of Bob inputs among all binary-output realizations of a fixed conflict graph is its edge biclique-cover number. Thus the graph correspondence is elevated from a construction principle to a complete classification and an exact input-compression theorem.

The paper is organized around three principal structural results. The first is this complete binary-output classification together with exact biclique-cover compression. The second is the arbitrary-output qubit impossibility theorem: regardless of the number of outputs, a perfect qubit support relation always has a perfect classical-bit realization. The third is a two-input qutrit impossibility theorem: with at most two Bob inputs, perfect qutrit support relations always admit perfect classical-trit realizations. The subsequent $G_{13}$, apex, Torpedo, few-input, and antidistinguishability constructions are presented as consequences, extensions, or boundary cases of these results.

We optimize a concrete realization of the binary obstruction. A vertex-test game uses all graph vertices as preparations but only selected vertices as Bob inputs. For the orthogonality graph of the Yu--Oh $13$-ray set \cite{YuOh2012}, denoted $G_{13}$ and studied graph-theoretically by Man\v{c}inska and Roberson \cite{MancinskaRoberson2016}, the symmetric realization has $C_3=59<Q_3=S=61$. Testing eight vertices gives $(X,Y,B)=(13,8,2)$ and $C_3=39<Q_3=S=40$. Since $\operatorname{bc}(G_{13})=8$, no binary-output realization of this conflict graph can use fewer Bob inputs. Exhaustive enumeration, supplemented by explicit coloring certificates, also shows that every one-vertex deletion of $G_{13}$ is $3$-colorable.

Apex extensions give binary perfect games in every dimension $d\ge3$, including a Bob-input-compressed family with $(X,Y,B)=(13+t,8+t,2)$. We also revisit the qutrit Torpedo relation. The qutrit SIC--MUB forbidden-output incidence pattern underlying the Torpedo relation also appears, up to relabeling, in the nonlocality, steering, and tomography inequality of Ref.~\cite{Huang2021}, and was subsequently formulated as a PM information-retrieval game in Ref.~\cite{Emeriau2022}. The kernels of its four forbidden-output linear functionals are the four directions of the affine plane $\F_3^2$. This geometry yields both the exact perfect-message cost and the general classical loss bound $L_{\mathrm{Tor}}\ge X-2M$. In particular, seven preparations are minimal for a qutrit-over-trit perfect separation within the Torpedo construction when all four Bob inputs $q\in\F_3\cup\{\infty\}$ are retained.

A separate question is how few Bob inputs can support a perfect same-dimensional advantage. We prove that two Bob inputs are impossible for qutrits, while a seven-preparation, three-input, three-output qutrit game attains $C_3=20<Q_3=S=21$. Motivated by recent quantum-coloring constructions \cite{Lalonde2026}, we also give a two-input, six-output game in dimension five with $C_5=39<Q_5=S=40$. Consequently, the smallest dimension admitting a two-input perfect same-dimensional separation is either four or five; we leave the four-dimensional case open.

\section{Exact-support prepare-and-measure games}
\label{sec:framework}

\subsection{Operational definition}

For a positive integer $n$, write $[n]:=\{1,\ldots,n\}$. Let $\calX$ and $\calY$ be finite preparation- and Bob-input sets, with
\begin{equation}
X:=|\calX|, \qquad Y:=|\calY|,
\end{equation}
and let $\calB$ be a finite output alphabet, with $B:=|\calB|$. We denote by $\mathcal D(\C^d)$ the set of density operators on $\C^d$, and by $\I_d$ the identity operator on that space; the subscript is omitted when the dimension is clear. A support game $\calG$ is specified by a promise set
\begin{equation}
\calP\subseteq\calX\times\calY
\end{equation}
and nonempty winning sets $\calA_{xy}\subseteq\calB$ for every $(x,y)\in\calP$. Its unnormalized score is
\begin{equation}
I_{\calG}[p] := \sum_{(x,y)\in\calP} \sum_{b\in\calA_{xy}}p(b|x,y),
\label{eq:general-score}
\end{equation}
with algebraic value
\begin{equation}
S_{\calG}=|\calP|.
\end{equation}
A behavior is perfect precisely when
\begin{equation}
p(b|x,y)=0 \quad \text{for all }(x,y)\in\calP\text{ and }b\notin\calA_{xy}.
\label{eq:support-condition}
\end{equation}

A $d$-dimensional quantum strategy consists of states $\rho_x\in\mathcal D(\C^d)$ and POVMs $\{M_{b|y}\}_{b\in\calB}$ such that
\begin{equation}
p(b|x,y)=\Tr(\rho_xM_{b|y}).
\end{equation}
A deterministic classical strategy with a $d$-level message is specified by an encoding $f:\calX\to[d]$ and a response function $g:[d]\times\calY\to\calB$. A general strategy with shared randomness is a convex combination of deterministic strategies. Since $I_{\calG}$ is linear in the behavior, its maximum is attained by a deterministic strategy. Moreover, if a shared-randomness strategy is perfect, every deterministic strategy appearing with nonzero weight is itself perfect. We denote by $Q_d(\calG)$ the quantum maximum at Hilbert-space dimension $d$, and by $C_d(\calG)$ the classical maximum for a $d$-level message.

We use the convention that an intersection over an empty family of winning sets equals $\calB$.

\begin{proposition}[Classical intersection criterion]
\label{prop:intersection}
A perfect classical strategy with a $d$-level message exists if and only if $\calX$ admits a partition into $r\le d$ nonempty message classes $\calX_1,\ldots,\calX_r$ such that
\begin{equation}
\bigcap_{\substack{x\in\calX_j\\(x,y)\in\calP}}
\calA_{xy}\neq\varnothing
\label{eq:intersection-criterion}
\end{equation}
for every $j\in[r]$ and every $y\in\calY$.
\end{proposition}

\begin{proof}
If a perfect strategy $(f,g)$ exists, discard the empty preimages of $f$ and label the remaining message classes $\calX_1,\ldots,\calX_r$, where $r\le d$. For fixed $(j,y)$, the output assigned to class $\calX_j$ belongs to every winning set in Eq.~\eqref{eq:intersection-criterion}. Conversely, choose one element of each intersection and use it as Bob's response for the corresponding class; define the response arbitrarily on unused message values. The resulting deterministic strategy wins every promised input pair.
\end{proof}

\subsection{Conflict hypergraphs and output cardinality}

A set $C\subseteq\calX$ is \emph{conflicting} if there exists a Bob input $y$ such that $(x,y)\in\calP$ for every $x\in C$ and
\begin{equation}
\bigcap_{x\in C}\calA_{xy}=\varnothing.
\end{equation}
No conflicting set can lie in one message class. The inclusion-minimal conflicting sets form a hypergraph on $\calX$, and Proposition~\ref{prop:intersection} identifies perfect classical strategies with proper colorings of this conflict hypergraph.

\begin{proposition}[Conflict-hypergraph rank bound]
\label{prop:output-rank}
If the output alphabet has cardinality $B=k$, every inclusion-minimal conflicting set contains at most $k$ preparations.
\end{proposition}

\begin{proof}
Fix a conflicting family of nonempty sets $\{\calA_{xy}:x\in C\}$ with empty intersection. For each $b\in\calB$, choose one set in the family that excludes $b$. The intersection of the at most $k$ selected sets is already empty. Hence an inclusion-minimal conflicting family has cardinality at most $k$.
\end{proof}

Thus binary-output support games have conflict hypergraphs of rank two. Define their conflict graph $\Gamma_{\calG}$ by
\begin{equation}
\begin{aligned}
x\sim_{\Gamma_{\calG}}x'
\quad\Longleftrightarrow\quad
\exists y:\;&(x,y),(x',y)\in\calP,\\
&\calA_{xy}\cap\calA_{x'y}=\varnothing.
\end{aligned}
\label{eq:conflict-graph}
\end{equation}

For a finite graph $H$, let $\chi(H)$ denote its chromatic number. Its complex orthogonal rank $\xi_{\C}(H)$ is the smallest integer $d$ for which there exist nonzero vectors $\{\ket{\psi_x}\}_{x\in V(H)}\subseteq\C^d$ satisfying
\begin{equation}
\braket{\psi_x}{\psi_{x'}}=0 \quad \text{whenever }\{x,x'\}\in E(H).
\end{equation}
The vectors may equivalently be chosen to be unit vectors. These graph parameters have direct antecedents in exact communication complexity. In particular, the minimum one-round quantum message dimension for graph-based promise-equality problems is characterized by complex orthogonal rank \cite{deWolf2001}; Ref.~\cite{Briet2015} gives a systematic treatment of the classical and quantum exact communication problems. The theorem below shows that every binary-output PM support game reduces to the same structure, even when no graph is specified a priori.

\begin{theorem}[Complete binary-output characterization]
\label{thm:binary-characterization}
Let $\calG$ be a binary-output support game with conflict graph $\Gamma_{\calG}$.
\begin{enumerate}
\item A perfect classical strategy with a $d$-level message exists if and only if
\begin{equation}
\chi(\Gamma_{\calG})\le d.
\label{eq:binary-classical}
\end{equation}
\item A perfect $d$-dimensional quantum strategy exists if and only if
\begin{equation}
\xi_{\C}(\Gamma_{\calG})\le d.
\label{eq:binary-quantum}
\end{equation}
\end{enumerate}
\end{theorem}

\begin{proof}
For $\calB=\{0,1\}$, every nonempty subset of $\calB$ is $\{0\}$, $\{1\}$, or $\{0,1\}$. A family of such sets has empty intersection if and only if it contains both a $\{0\}$ set and a $\{1\}$ set. Hence all classical conflicts are pairwise, and Proposition~\ref{prop:intersection} is exactly the proper-coloring condition in Eq.~\eqref{eq:binary-classical}.

For the forward quantum implication, let $x\sim_{\Gamma_{\calG}}x'$. For some $y$, after possibly exchanging the output labels,
\begin{equation}
\calA_{xy}=\{0\}, \qquad \calA_{x'y}=\{1\}.
\end{equation}
For positive operators $\rho$ and $M$,
\begin{equation}
\Tr(\rho M)=0 \quad\Longrightarrow\quad \supp(\rho)\subseteq\ker M.
\label{eq:positive-zero-support}
\end{equation}
Indeed, $\rho^{1/2}M\rho^{1/2}$ is positive with zero trace and is therefore zero. Equivalently, $M^{1/2}\rho^{1/2}=0$, which implies $\operatorname{ran}\rho^{1/2}\subseteq\ker M$ and hence $\supp(\rho)\subseteq\ker M$. Applying Eq.~\eqref{eq:positive-zero-support} to the losing effects, perfection implies
\begin{equation}
\supp(\rho_x)\subseteq\ker M_{1|y}, \qquad \supp(\rho_{x'})\subseteq\ker M_{0|y}.
\end{equation}
Since $M_{0|y}=\I-M_{1|y}$, these kernels are respectively the eigenspaces of the Hermitian operator $M_{1|y}$ with eigenvalues $0$ and $1$, and are therefore orthogonal. Choosing one unit vector from the support of each $\rho_x$ gives a $d$-dimensional orthogonal representation of $\Gamma_{\calG}$.

Conversely, suppose that $\{\ket{\psi_x}\}_{x\in\calX}$ is a $d$-dimensional orthogonal representation of $\Gamma_{\calG}$, and let $\rho_x=\ket{\psi_x}\!\bra{\psi_x}$. For each Bob input $y$, define, with the convention $\operatorname{span}\varnothing=\{0\}$,
\begin{equation}
U_{1|y} := \operatorname{span} \left\{ \ket{\psi_x}: (x,y)\in\calP,\ \calA_{xy}=\{1\} \right\}.
\end{equation}
Let $\Pi_{1|y}$ be the orthogonal projector onto $U_{1|y}$ and set
\begin{equation}
M_{1|y}=\Pi_{1|y}, \qquad M_{0|y}=\I-\Pi_{1|y}.
\end{equation}
If $\calA_{xy}=\{1\}$, then $\ket{\psi_x}\in U_{1|y}$ and hence $p(1|x,y)=1$. If $\calA_{xy}=\{0\}$, then $x$ is adjacent to every $x'$ for which $\calA_{x'y}=\{1\}$. Thus $\ket{\psi_x}\perp U_{1|y}$ and $p(0|x,y)=1$. Winning sets $\calA_{xy}=\{0,1\}$ impose no constraint. This defines a perfect $d$-dimensional quantum strategy.
\end{proof}

\begin{corollary}[Binary perfect-separation criterion]
\label{cor:binary-gap}
A binary-output support game $\calG$ exhibits a perfect same-dimensional separation in dimension $d$ if and only if
\begin{equation}
\xi_{\C}(\Gamma_{\calG})\le d<\chi(\Gamma_{\calG}).
\label{eq:orthogonal-chromatic-gap}
\end{equation}
\end{corollary}

\begin{proof}
This follows directly from Theorem~\ref{thm:binary-characterization}.
\end{proof}

For disjoint vertex sets $L$ and $R$, let $K_{L,R}$ denote the complete bipartite graph with parts $L$ and $R$. For a graph $G$, let $\operatorname{bc}(G)$ denote its \emph{edge biclique-cover number}, namely the minimum number of complete bipartite subgraphs whose edge sets cover $E(G)$.

\begin{proposition}[Minimum Bob-input count]
\label{prop:biclique-cover}
Among all binary-output support games with conflict graph $G$, the minimum possible number of Bob inputs is
\begin{equation}
Y_{\min}(G)=\operatorname{bc}(G).
\label{eq:min-receiver-biclique}
\end{equation}
\end{proposition}

\begin{proof}
For a Bob input $y$, define
\begin{align}
L_y&:=\{x:(x,y)\in\calP,\ \calA_{xy}=\{0\}\},\nonumber\\
R_y&:=\{x:(x,y)\in\calP,\ \calA_{xy}=\{1\}\}.
\end{align}
The conflict edges generated by $y$ are exactly the edges of the complete bipartite graph with parts $L_y$ and $R_y$. Therefore the Bob inputs of any binary-output realization provide a biclique cover of its conflict graph, and hence $Y\ge\operatorname{bc}(G)$.

Conversely, let $\{K_{L_y,R_y}\}_{y=1}^{\operatorname{bc}(G)}$ be an edge biclique cover of $G$. For input $y$, promise precisely the preparations in $L_y\cup R_y$, assign winning set $\{0\}$ to $L_y$ and $\{1\}$ to $R_y$, and leave all other pairs unpromised. The resulting binary-output support game has conflict graph exactly $G$ and uses $\operatorname{bc}(G)$ Bob inputs.
\end{proof}

For $d=2$, every graph with a two-dimensional orthogonal representation is bipartite: along a path, vectors at even distance are collinear and vectors at odd distance are orthogonal to them. Thus Corollary~\ref{cor:binary-gap} already excludes binary-output qubit separations. We next prove the stronger result for arbitrary finite output alphabets.

\section{No perfect qubit advantage over a classical bit}
\label{sec:no-qubit}

\begin{theorem}[Arbitrary-output qubit impossibility]
\label{thm:no-qubit}
For every finite support game $\calG$,
\begin{equation}
Q_2(\calG)=S_{\calG} \quad\Longrightarrow\quad C_2(\calG)=S_{\calG}.
\end{equation}
\end{theorem}

\begin{proof}
Assume that a perfect qubit strategy exists. We first replace every preparation by a pure state without changing any zero-probability constraint. A full-rank qubit state has strictly positive overlap with every nonzero positive operator. Hence, if $\rho_x$ is full rank, every losing effect for $x$ is zero and remains forbidden after an arbitrary pure-state replacement. Writing $\bm\sigma:=(\sigma_x,\sigma_y,\sigma_z)$ for the vector of Pauli matrices, we may therefore express each preparation as
\begin{equation}
\rho_x=\frac{1}{2}(\I+\bm r_x\cdot\bm\sigma), \qquad \|\bm r_x\|=1.
\end{equation}

Because the set of Bloch vectors is finite, choose a unit vector $\bm n$ such that $\bm n\cdot\bm r_x\neq0$ for every $x$. Alice sends the bit
\begin{equation}
c(x)=
\begin{cases}
0,&\bm n\cdot\bm r_x>0,\\
1,&\bm n\cdot\bm r_x<0.
\end{cases}
\label{eq:hemisphere-bit}
\end{equation}

Fix $y$ and a nonempty hemisphere class $X_c=\{x:c(x)=c\}$ containing at least one preparation promised with $y$. Suppose that no output is winning for every promised pair $(x,y)$ with $x\in X_c$. Explicitly, this means that for every output $b$ there exists a promised preparation $x_b\in X_c$, with $(x_b,y)\in\calP$, such that $b\notin\calA_{x_by}$. Restricting to nonzero effects $M_{b|y}$, perfection gives
\begin{equation}
\Tr(\rho_{x_b}M_{b|y})=0.
\end{equation}
A nonzero positive qubit operator annihilating a pure state has rank one, and therefore
\begin{equation}
M_{b|y} = \lambda_b\frac{\I-\bm r_{x_b}\cdot\bm\sigma}{2}, \qquad \lambda_b>0.
\end{equation}
Summing over the nonzero effects and comparing the identity and Pauli coefficients in $\sum_bM_{b|y}=\I$ yields
\begin{equation}
\sum_{b:M_{b|y}\neq0}\lambda_b=2, \qquad \sum_{b:M_{b|y}\neq0}\lambda_b\bm r_{x_b}=0.
\label{eq:hemisphere-zero}
\end{equation}
All vectors in this sum lie in the same open hemisphere. Taking the scalar product with $\bm n$ gives a strictly positive number for $c=0$ and a strictly negative number for $c=1$, contradicting Eq.~\eqref{eq:hemisphere-zero}.

Thus, whenever $X_c$ contains a preparation promised with $y$, there exists an output $b(c,y)$ that wins for all such preparations. If no preparation in $X_c$ is promised with $y$, Bob chooses $b(c,y)$ arbitrarily. Bob outputs $b(c,y)$ after receiving $c$, which defines a perfect classical-bit strategy.
\end{proof}

Theorem~\ref{thm:no-qubit} concerns exact support rather than full statistics. General qubit PM behaviors can require four classical messages to reproduce exactly \cite{Renner2023,Schlosser2026}; nevertheless, every perfectly realizable qubit support relation needs only two messages. Since a one-output game is trivial, dimension three and two outputs are the first possible parameters for a nontrivial perfect same-dimensional separation.

\section{Few Bob inputs}
\label{sec:few-inputs}

\subsection{Two-input qutrit impossibility}

We first isolate the linear-algebraic fact needed below. A finite family $\Phi=\{\phi_i\}\subset\C^3$ is a Parseval frame if
\begin{equation}
\sum_i \ket{\phi_i}\!\bra{\phi_i}=\I_3.
\end{equation}
This is the standard finite-dimensional Parseval-frame normalization \cite{Christensen2016}.

\begin{lemma}[Rank-three Parseval transversal lemma]
\label{lem:parseval-transversal}
Let $\Phi$ and $\Psi$ be finite Parseval frames in $\C^3$. Then there exist bases
\begin{equation}
E=(e_1,e_2,e_3),
\qquad
F=(f_1,f_2,f_3),
\end{equation}
formed by vectors from $\Phi$ and $\Psi$, respectively, that can be ordered so that every transversal
\begin{equation}
\{w_1,w_2,w_3\},
\qquad
w_i\in\{e_i,f_i\},
\label{eq:transversal-basis}
\end{equation}
is a basis of $\C^3$.
\end{lemma}

A proof is given in Appendix~\ref{app:parseval-transversal}. The proof uses a rank-three classification and does not automatically extend to higher dimensions; in particular, the dimension-four Parseval-frame analogue is not established here.

\begin{theorem}[Two-input qutrit impossibility]
\label{thm:no-qutrit-two-input}
Let $\calG$ be a finite support game with at most two Bob inputs. Then
\begin{equation}
Q_3(\calG)=S_{\calG}
\quad\Longrightarrow\quad
C_3(\calG)=S_{\calG}.
\label{eq:no-qutrit-two-input}
\end{equation}
This holds for arbitrary finite output alphabets, arbitrary qutrit POVMs, mixed preparations, and arbitrary promise sets.
\end{theorem}

\begin{proof}
It suffices to consider exactly two Bob inputs: if the game has only one, append a second input for which every pair is unpromised. Assume a perfect qutrit strategy. We may regard every unpromised pair $(x,y)$ as having the full output alphabet as its winning set. Spectrally refine the two POVMs into rank-one positive terms,
\begin{equation}
M_{a|1}=\sum_r \ket{u_{a,r}}\!\bra{u_{a,r}},
\qquad
M_{b|2}=\sum_s \ket{v_{b,s}}\!\bra{v_{b,s}}.
\end{equation}
Completeness makes the two collections of nonzero refined vectors Parseval frames. By Lemma~\ref{lem:parseval-transversal}, choose and pair bases $u_1,u_2,u_3$ and $v_1,v_2,v_3$ such that every transversal is a basis. Let $a_i$ and $b_i$ denote the original coarse outcomes containing $u_i$ and $v_i$, respectively.

Use three classical messages $i\in\{1,2,3\}$ and let Bob output
\begin{equation}
g(i,1)=a_i,
\qquad
g(i,2)=b_i.
\label{eq:qutrit-two-input-decoder}
\end{equation}
Suppose that a preparation $x$ cannot use any of these three messages. For every $i$, at least one of $a_i$ at $y=1$ and $b_i$ at $y=2$ is losing for $x$. Choose $w_i=u_i$ in the former case and $w_i=v_i$ in the latter. If, for example, $a_i$ is losing, perfection gives
\begin{equation}
0=\Tr(\rho_xM_{a_i|1})
\ge \bra{u_i}\rho_x\ket{u_i},
\end{equation}
so positivity implies $u_i\perp\supp(\rho_x)$; the same holds for a losing $b_i$. Thus the selected transversal $\{w_1,w_2,w_3\}$ is orthogonal to $\supp(\rho_x)$. Lemma~\ref{lem:parseval-transversal} says that this transversal spans $\C^3$, forcing $\rho_x=0$, a contradiction. Hence every preparation can use at least one message in Eq.~\eqref{eq:qutrit-two-input-decoder}, giving a perfect classical trit strategy.
\end{proof}

\subsection{A seven-preparation qutrit game with three Bob inputs}
\label{subsec:seven-three-three}

The preceding lower bound on the number of Bob inputs is tight. Consider the seven pure qutrit preparations
\begin{equation}
\begin{aligned}
\psi_1&=(1,0,0),&
\psi_2&=(0,1,0),\\
\psi_3&=(0,0,1),&
\psi_4&=\frac{1}{\sqrt3}(1,1,1),\\
\psi_5&=\frac{1}{\sqrt3}(1,1,-1),&
\psi_6&=\frac{1}{\sqrt3}(1,-1,1),\\
\psi_7&=\frac{1}{\sqrt3}(1,-1,-1).&&
\end{aligned}
\label{eq:seven-qutrit-states}
\end{equation}
and the three projective measurements whose ordered orthonormal bases are
\begin{equation}
\begin{aligned}
\mathcal M_1&=\left\{
(0,1,0),\frac{(1,0,-1)}{\sqrt2},\frac{(1,0,1)}{\sqrt2}
\right\},\\
\mathcal M_2&=\left\{
\frac{(0,1,-1)}{\sqrt2},\frac{(0,1,1)}{\sqrt2},(1,0,0)
\right\},\\
\mathcal M_3&=\left\{
\frac{(1,1,0)}{\sqrt2},(0,0,1),\frac{(1,-1,0)}{\sqrt2}
\right\}.
\end{aligned}
\label{eq:seven-qutrit-measurements}
\end{equation}
For every $(x,y)$, declare precisely the outcomes with nonzero Born probability to be winning. The resulting winning sets are
\begin{equation}
\begin{array}{c|ccc}
 x & y=1 & y=2 & y=3\\
\hline
1&\{2,3\}&\{3\}&\{1,3\}\\
2&\{1\}&\{1,2\}&\{1,3\}\\
3&\{2,3\}&\{1,2\}&\{2\}\\
4&\{1,3\}&\{2,3\}&\{1,2\}\\
5&\{1,2\}&\{1,3\}&\{1,2\}\\
6&\{1,3\}&\{1,3\}&\{2,3\}\\
7&\{1,2\}&\{2,3\}&\{2,3\}
\end{array}
\label{eq:seven-qutrit-winning-sets}
\end{equation}
with all $21$ input pairs promised.

\begin{proposition}[Seven-preparation three-input qutrit game]
\label{prop:seven-three-three}
The game in Eqs.~\eqref{eq:seven-qutrit-states}--\eqref{eq:seven-qutrit-winning-sets} has
\begin{equation}
(X,Y,B)=(7,3,3),
\qquad
C_3=20,
\qquad
Q_3=S=21.
\label{eq:seven-three-three-values}
\end{equation}
\end{proposition}

\begin{proof}
Quantum perfection follows directly from the definition of the winning sets. For a deterministic classical message, Bob fixes a response triple $s=(b_1,b_2,b_3)\in[3]^3$. The decoder triples compatible with preparations $1,2,3$ belong respectively to
\begin{equation}
\begin{aligned}
D_1&=\{2,3\}\times\{3\}\times\{1,3\},\\
D_2&=\{1\}\times\{1,2\}\times\{1,3\},\\
D_3&=\{2,3\}\times\{1,2\}\times\{2\}.
\end{aligned}
\end{equation}
These three sets are pairwise disjoint, so a perfect classical strategy already requires three distinct messages for preparations $1,2,3$. Directly from Eq.~\eqref{eq:seven-qutrit-winning-sets}, every decoder triple in $D_i$ is compatible with at most one of preparations $4,5,6,7$. Hence three messages can cover at most three of those four remaining preparations, and perfection is impossible. Since deterministic scores are integral, $C_3\le20$.

The three decoder triples
\begin{equation}
(1,1,1),\qquad(1,2,2),\qquad(3,3,3)
\end{equation}
achieve score $20$: assign preparations $2,5$ to $(1,1,1)$, preparations $3,4,7$ to $(1,2,2)$, and preparations $1,6$ to $(3,3,3)$. The only lost input pair is $(x,y)=(3,1)$. Thus $C_3=20$.
\end{proof}

Combining Theorem~\ref{thm:no-qutrit-two-input} and Proposition~\ref{prop:seven-three-three} gives
\begin{equation}
Y_{\min}^{\mathrm{perfect}}(d=3)=3,
\label{eq:qutrit-Y-minimum}
\end{equation}
where the minimum is over arbitrary finite output alphabets and support games.

The vectors in Eqs.~\eqref{eq:seven-qutrit-states} and \eqref{eq:seven-qutrit-measurements} lie in the familiar Yu--Oh qutrit geometry \cite{YuOh2012}; the result is therefore a compact support-game realization of known qutrit rays rather than a new ray configuration. More generally, recent work on pseudocontexts identifies systematic operator identities of the form $\sum_{v\in A_y}P_v=T$ \cite{NavaraSvozil2026}. For invertible $T$, the common whitening $T^{-1/2}P_vT^{-1/2}$ turns every such decomposition into a POVM. This common-operator viewpoint is a useful search principle for perfect PM games. The construction above uses the particularly simple case $T=\I$, namely ordinary orthonormal contexts; it does not rely on a nontrivial pseudocontext identity.

\subsection{Two Bob inputs in higher dimension}
\label{subsec:two-input-d5}

The qutrit impossibility is dimension-specific. We next give a two-input perfect same-dimensional separation in dimension five. The construction was found while exploring structures motivated by recent quantum-coloring results \cite{Lalonde2026}; the verification below is self-contained.

Let $e_1,\ldots,e_5$ be the standard basis of $\C^5$ and $\mathbf 1=\sum_{i=1}^5e_i$. Bob's first six-outcome POVM is
\begin{equation}
M_{i|1}=\ket{e_i}\!\bra{e_i}\quad(i=1,\ldots,5),
\qquad
M_{6|1}=0.
\end{equation}
For the second POVM define
\begin{equation}
g_i=\frac13\mathbf1-e_i\quad(i=1,\ldots,5),
\qquad
g_6=\frac13\mathbf1,
\label{eq:d5-g-vectors}
\end{equation}
and set $M_{i|2}=\ket{g_i}\!\bra{g_i}$. A direct expansion gives
\begin{equation}
\sum_{i=1}^6\ket{g_i}\!\bra{g_i}=\I_5.
\label{eq:d5-parseval}
\end{equation}

Alice has ten pair preparations and ten triple preparations,
\begin{equation}
\ket{\psi^-_{ij}}=\frac{e_i-e_j}{\sqrt2}
\quad(1\le i<j\le5),
\end{equation}
\begin{equation}
\ket{\psi^+_T}=\frac1{\sqrt3}\sum_{i\in T}e_i
\quad(T\in\tbinom{[5]}3).
\end{equation}
Their supports are
\begin{equation}
\calA_{ij,1}=\calA_{ij,2}=\{i,j\},
\label{eq:d5-pair-supports}
\end{equation}
and
\begin{equation}
\calA_{T,1}=T,
\qquad
\calA_{T,2}=T^c\cup\{6\}.
\label{eq:d5-triple-supports}
\end{equation}
Indeed, pair states have zero coordinate sum, while a triple state has equal nonzero overlap with precisely the outcomes in Eq.~\eqref{eq:d5-triple-supports}.

\begin{theorem}[A two-input perfect separation in dimension five]
\label{thm:d5-two-input}
The support game in Eqs.~\eqref{eq:d5-pair-supports} and \eqref{eq:d5-triple-supports} has
\begin{equation}
(X,Y,B)=(20,2,6),
\qquad
C_5=39,
\qquad
Q_5=S=40.
\label{eq:d5-two-input-values}
\end{equation}
Its exact classical zero-error message cost is six.
\end{theorem}

\begin{proof}
Equation~\eqref{eq:d5-parseval} and the support calculation above give a perfect five-dimensional quantum strategy, so $Q_5=S=40$.

A deterministic classical message is a decoder pair $(a,b)$, giving Bob's outputs at $y=1,2$. Duplicate decoder pairs are redundant, so assume the decoder pairs are distinct. A decoder pair with first component $6$ covers no preparation, since output $6$ is never winning at $y=1$, and may therefore be discarded. Call $(i,i)$, $i\in[5]$, diagonal. A diagonal message covers all four pair preparations containing $i$, whereas a non-diagonal $(i,j)$ with $i,j\in[5]$ covers only the pair $\{i,j\}$, and $(i,6)$ covers no pair preparation. Suppose at most five messages were perfect and let $r$ be the number of diagonal messages. If the corresponding diagonal labels form a set $D\subset[5]$, every pair contained in $[5]\setminus D$ requires its own non-diagonal message. Therefore
\begin{equation}
r+\binom{5-r}{2}\le5,
\end{equation}
which forces $r\ge2$.

If $r=2$, all three remaining messages are needed for the three pairs inside the complementary three-element set $C=[5]\setminus D$. Both outputs of each of these messages lie in $C$, so none covers the triple preparation $T=C$; diagonal messages cover no triple because Eq.~\eqref{eq:d5-triple-supports} would require the same label to lie in both $T$ and $T^c$. This is impossible.

If $r\ge3$, at most two non-diagonal messages remain to cover all ten triple preparations. A message $(a,6)$ covers the six triples containing $a$, while an ordinary non-diagonal message $(a,b)$ covers only the three triples with $a\in T$ and $b\notin T$. Two ordinary messages therefore cannot cover all ten triples; one message of type $(a,6)$ together with any other leaves a triple uncovered, and two messages $(a,6),(c,6)$ leave the triple contained in $[5]\setminus\{a,c\}$ uncovered. Thus five messages can never be perfect.

Consequently $C_5\le39$. The five decoder pairs
\begin{equation}
(1,1),\ (2,2),\ (3,3),\ (4,5),\ (5,6)
\end{equation}
attain $39$. Indeed, the first three diagonal messages cover every pair preparation touching $1,2,$ or $3$, while $(4,5)$ covers the remaining pair $\{4,5\}$. For triples, $(5,6)$ covers every triple containing $5$, and $(4,5)$ covers every triple containing $4$ but not $5$; the only triple not perfectly covered is $T=\{1,2,3\}$, for which a diagonal decoder loses exactly the $y=2$ condition. Finally, the six decoder pairs
\begin{equation}
(1,1),\ (1,2),\ (2,3),\ (3,6),\ (4,4),\ (5,5)
\end{equation}
are perfect: the diagonal messages with labels $1,4,5$ cover all pair preparations except $\{2,3\}$, which is covered by $(2,3)$; $(3,6)$ covers every triple containing $3$, $(2,3)$ covers every remaining triple containing $2$, and the sole remaining triple $\{1,4,5\}$ is covered by $(1,2)$. Hence the zero-error classical message cost is exactly six.
\end{proof}

Define the smallest dimension permitting a two-Bob-input perfect same-dimensional separation by
\begin{equation}
d_{\min}^{(Y=2)}
:=
\min\left\{d:
\substack{\exists\,\calG\text{ with }Y=2,\\
C_d(\calG)<Q_d(\calG)=S_{\calG}}
\right\}.
\end{equation}
Dimension one is trivially equivalent to a one-valued classical message. Together with Theorems~\ref{thm:no-qubit} and \ref{thm:no-qutrit-two-input}, this gives $d_{\min}^{(Y=2)}\ge4$, while Theorem~\ref{thm:d5-two-input} gives $d_{\min}^{(Y=2)}\le5$. Therefore
\begin{equation}
d_{\min}^{(Y=2)}\in\{4,5\}.
\label{eq:y2-dimension-window}
\end{equation}
Whether a four-dimensional two-input perfect same-dimensional separation exists remains open.

\section{Graph realizations and Bob-input compression}
\label{sec:graph-construction}

Let $G=(V,E)$ be a finite simple graph. For $v\in V$, let $N_G(v)$ denote its neighborhood and $\deg_G(v):=|N_G(v)|$ its degree; we omit the subscript $G$ when the graph is unambiguous. Let $T\subseteq V$ be a set of tested vertices. Alice receives $x\in V$, Bob receives $y\in T$, and Bob outputs $b\in\{0,1\}$. Only $x=y$ and $x\sim y$ are promised: Bob should output $1$ on the diagonal and $0$ on adjacent vertices. The corresponding functional is
\begin{equation}
I_{G,T} := \sum_{y\in T}p(1|y,y) + \sum_{y\in T}\sum_{x\in N(y)}p(0|x,y),
\label{eq:graph-cover-functional}
\end{equation}
with algebraic value
\begin{equation}
S_{G,T}=|T|+\sum_{y\in T}\deg(y).
\label{eq:graph-cover-S}
\end{equation}
We write $Q_d(G,T)$ and $C_d(G,T)$ for the corresponding $d$-dimensional quantum and classical values.

\begin{theorem}[Vertex-test graph game]
\label{thm:graph-cover}
If $\xi_{\C}(G)\le d$, then
\begin{equation}
Q_d(G,T)=S_{G,T}
\end{equation}
for every $T\subseteq V$. For a fixed classical encoding $f:V\to[d]$, a perfect response function exists if and only if
\begin{equation}
f(x)\neq f(y) \quad \text{for every }y\in T\text{ and }x\in N(y).
\label{eq:test-edge-separation}
\end{equation}
In particular, if $T$ is a vertex cover, then
\begin{equation}
C_d(G,T)=S_{G,T} \quad\Longleftrightarrow\quad \chi(G)\le d.
\label{eq:cover-classical-iff}
\end{equation}
\end{theorem}

\begin{proof}
Let $\{\ket{\psi_x}\}$ be a $d$-dimensional orthogonal representation. Alice prepares
\begin{equation}
\rho_x=\ket{\psi_x}\!\bra{\psi_x},
\end{equation}
and Bob uses
\begin{equation}
M_{1|y}=\ket{\psi_y}\!\bra{\psi_y}, \qquad M_{0|y}=\I-M_{1|y}.
\end{equation}
Then $p(1|y,y)=1$ and $p(0|x,y)=1$ whenever $x\sim y$.

For a deterministic classical strategy, saturating the diagonal term for $y$ requires $g(f(y),y)=1$, whereas saturating the term associated with $x\sim y$ requires $g(f(x),y)=0$. These demands are compatible exactly when Eq.~\eqref{eq:test-edge-separation} holds. Conversely, under this condition Bob outputs $1$ precisely on message $f(y)$ and $0$ otherwise. If $T$ is a vertex cover, every edge is tested and Eq.~\eqref{eq:test-edge-separation} is precisely a proper $d$-coloring.
\end{proof}

For an arbitrary color assignment $f:V\to[d]$, a tested vertex $y\in T$ is called \emph{conflicting under $f$} if it has a neighbor of the same color, i.e., if there exists $x\in N(y)$ such that $f(x)=f(y)$. Define
\begin{equation}
D_T(f) := \left| \left\{ y\in T:\exists x\in N(y),\ f(x)=f(y) \right\} \right|
\end{equation}
to be the number of conflicting tested vertices. A tested vertex is
\emph{locally proper under $f$} if it is not conflicting under $f$. We write
\begin{equation}
\ell_T(f):=|T|-D_T(f)
\end{equation}
for the number of locally proper tested vertices. The $d$-color deficit is
\begin{equation}
\delta_d(G,T):=\min_{f:V\to[d]}D_T(f).
\label{eq:deficit}
\end{equation}

\begin{proposition}[Exact classical value]
\label{prop:graph-classical-value}
For every graph $G$, tested set $T$, and positive integer $d$, the classical value for a $d$-level message is
\begin{equation}
C_d(G,T)=S_{G,T}-\delta_d(G,T).
\label{eq:graph-value-deficit}
\end{equation}
\end{proposition}

\begin{proof}
Fix $y\in T$. For every message other than $f(y)$, Bob outputs $0$ and wins all associated edge terms. In the class $f(y)$, output $1$ wins the diagonal term, whereas output $0$ wins all same-colored edge terms. The block value is $\deg(y)+1$ if $y$ has no same-colored neighbor and $\deg(y)$ otherwise. Summing over $y$ and optimizing $f$ gives Eq.~\eqref{eq:graph-value-deficit}.
\end{proof}

The choice $T=V$ gives the symmetric graph functional
\begin{equation}
I_G= \sum_{x\in V}p(1|x,x) + \sum_{\{x,y\}\in E} \bigl[p(0|x,y)+p(0|y,x)\bigr].
\label{eq:symmetric-graph-functional}
\end{equation}
For $T=V(G)$, we abbreviate $I_G:=I_{G,V(G)}$, $S_G:=S_{G,V(G)}=|V(G)|+2|E(G)|$, $Q_d(G):=Q_d(G,V(G))$, and $C_d(G):=C_d(G,V(G))$. If $T$ is a vertex cover, choosing $T\subsetneq V$ reduces Bob's input alphabet while preserving the full conflict graph. For a general tested set, the resulting conflict graph contains precisely the edges of $G$ incident on $T$.

\section{Binary-output qutrit games from \texorpdfstring{$G_{13}$}{G13}}
\label{sec:G13}

We now apply the binary classification and biclique-compression theorem to a canonical qutrit orthogonality graph. The resulting $G_{13}$ game is the flagship finite example: it gives an explicit perfect qutrit-over-trit separation and permits an exact optimization of Bob's input alphabet.

\subsection{The orthogonality graph}

Consider the thirteen rays generated by
\begin{align}
\bm v_1&=(1,0,0),&
\bm v_2&=(0,1,0),&
\bm v_3&=(0,0,1),\nonumber\\
\bm v_4&=(0,1,1),&
\bm v_5&=(0,1,-1),&
\bm v_6&=(1,0,1),\nonumber\\
\bm v_7&=(1,0,-1),&
\bm v_8&=(1,1,0),&
\bm v_9&=(1,-1,0),\nonumber\\
\bm v_{10}&=(1,1,1),&
\bm v_{11}&=(1,1,-1),&
\bm v_{12}&=(1,-1,1),\nonumber\\
\bm v_{13}&=(-1,1,1).&&
\label{eq:G13-vectors}
\end{align}
Let $G_{13}$ connect orthogonal rays. It has $13$ vertices and $24$ edges. The vectors $\bm v_1,\bm v_2,\bm v_3$ form a triangle, so $\xi_{\C}(G_{13})\ge3$, while the normalized vectors in Eq.~\eqref{eq:G13-vectors} give the reverse inequality. Thus
\begin{equation}
\xi_{\C}(G_{13})=3.
\end{equation}
This is the orthogonality graph of the Yu--Oh qutrit ray set \cite{YuOh2012}. Man\v{c}inska and Roberson denoted it by $G_{13}$ and established, in particular, that $\chi(G_{13})=4$ \cite{MancinskaRoberson2016}.

For later use, consider the trit encoding
\begin{equation}
\bm f_{\mathrm{near}} = (1,1,2,2,3,2,3,3,1,2,3,1,1),
\label{eq:near-coloring}
\end{equation}
whose unique monochromatic edge is $\{1,2\}$. Figure~\ref{fig:G13} displays the graph and the minimum tested set used below.

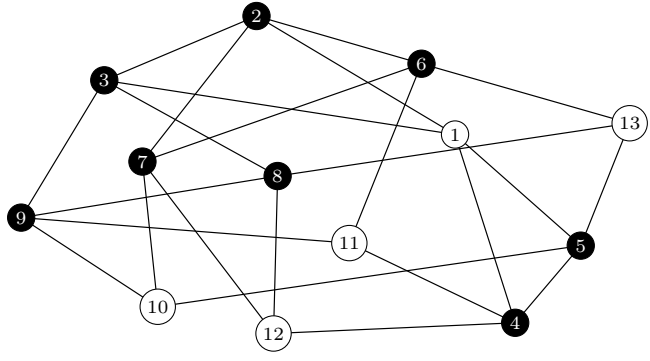
\begin{figure*}[t]
\centering
\begin{tikzpicture}[
  x=1.15cm,y=1.05cm,
  edge/.style={line width=0.45pt},
  tested/.style={circle,draw,fill=black,text=white,inner sep=1.5pt,font=\scriptsize},
  untested/.style={circle,draw,fill=white,text=black,inner sep=1.5pt,font=\scriptsize}
]
\coordinate (v1) at (4.991,2.507);
\coordinate (v2) at (2.707,4.000);
\coordinate (v3) at (0.956,3.190);
\coordinate (v4) at (5.683,0.135);
\coordinate (v5) at (6.437,1.107);
\coordinate (v6) at (4.605,3.399);
\coordinate (v7) at (1.395,2.166);
\coordinate (v8) at (2.950,1.984);
\coordinate (v9) at (0.000,1.460);
\coordinate (v10) at (1.571,0.337);
\coordinate (v11) at (3.775,1.139);
\coordinate (v12) at (2.903,0.000);
\coordinate (v13) at (7.000,2.649);
\foreach \a/\b in {
1/2,1/3,1/4,1/5,2/3,2/6,2/7,3/8,3/9,
4/5,4/11,4/12,5/10,5/13,6/7,6/11,6/13,
7/10,7/12,8/9,8/12,8/13,9/10,9/11}
  \draw[edge] (v\a)--(v\b);
\foreach \v in {2,3,4,5,6,7,8,9}
  \node[tested] at (v\v) {\v};
\foreach \v in {1,10,11,12,13}
  \node[untested] at (v\v) {\v};
\end{tikzpicture}
\caption{An abstract drawing of the orthogonality graph $G_{13}$. Vertex positions are chosen solely for readability and do not represent the coordinates of the rays in Eq.~\eqref{eq:G13-vectors}. Filled vertices form the tested set $T_8=\{2,3,4,5,6,7,8,9\}$. Its five untested vertices form a maximum independent set, so $T_8$ is a minimum vertex cover.}
\label{fig:G13}
\end{figure*}

\subsection{Symmetric game}

Taking $T=V(G_{13})$ gives
\begin{equation}
S_{G_{13}}=13+2\cdot24=61.
\end{equation}
The qutrit strategy of Theorem~\ref{thm:graph-cover} is perfect. Every trit encoding $f$ has a monochromatic edge because $\chi(G_{13})=4$. Both endpoints of that edge are conflicting under $f$, so $\delta_3(G_{13},V(G_{13}))\ge2$. The encoding in Eq.~\eqref{eq:near-coloring} has a unique monochromatic edge and therefore attains equality. Proposition~\ref{prop:graph-classical-value} yields
\begin{equation}
C_3(G_{13})=59, \qquad Q_3(G_{13})=S_{G_{13}}=61.
\label{eq:G13-values}
\end{equation}

\subsection{Optimal Bob-input compression}

The set
\begin{equation}
T_8:=\{2,3,4,5,6,7,8,9\}
\label{eq:T8}
\end{equation}
is a vertex cover because its complement $\{1,10,11,12,13\}$ is independent. Every vertex in $T_8$ has degree four, and hence
\begin{equation}
S_{G_{13},T_8}=8+8\cdot4=40.
\end{equation}
Every trit encoding has a monochromatic edge, and at least one endpoint is tested; hence $\delta_3(G_{13},T_8)\ge1$. In the near-coloring \eqref{eq:near-coloring}, the only monochromatic edge is $\{1,2\}$ and only vertex $2$ is tested. Thus
\begin{equation}
C_3(G_{13},T_8)=39, \qquad Q_3(G_{13},T_8)=S_{G_{13},T_8}=40.
\label{eq:G13-compressed-values}
\end{equation}
The resulting scenario is
\begin{equation}
(X,Y,B)=(13,8,2).
\end{equation}

The value $Y=8$ is minimal not only among vertex-test realizations but among all binary-output support games with conflict graph $G_{13}$. Indeed, direct verification from Eq.~\eqref{eq:G13-edges} shows that $G_{13}$ contains no $4$-cycle. Every complete bipartite subgraph is therefore a star, since a biclique with at least two vertices in each part would contain a $4$-cycle. Writing $\alpha(G)$ and $\tau(G)$ for the independence number and vertex-cover number, respectively, selecting the centers of any star cover gives a vertex cover; conversely, the full stars centered at the vertices of any vertex cover cover every edge. Thus covering $E(G_{13})$ by bicliques is equivalent to covering it by vertex-centered stars, and hence
\begin{equation}
\operatorname{bc}(G_{13}) = \tau(G_{13}) = 13-\alpha(G_{13}) =8.
\label{eq:G13-biclique-cover}
\end{equation}
Proposition~\ref{prop:biclique-cover} then proves that every binary-output realization of the $G_{13}$ conflict graph requires at least eight Bob inputs, while the game above attains this minimum.

Within the vertex-test realization, one can additionally classify every minimum tested set. For a tested set $T\subseteq V(G_{13})$, define
\begin{equation}
H_T := \bigl(V(G_{13}),\{\{u,v\}\in E(G_{13}):u\in T\ \text{or}\ v\in T\}\bigr).
\label{eq:tested-edge-graph}
\end{equation}
Thus $H_T$ contains precisely the edges of $G_{13}$ incident on $T$. Since $H_T$ is a subgraph of $G_{13}$, the vectors in Eq.~\eqref{eq:G13-vectors} provide a three-dimensional orthogonal representation of it. Theorem~\ref{thm:binary-characterization} therefore implies that the restricted game remains a perfect qutrit-over-trit separation exactly when $H_T$ is not $3$-colorable. Exact enumeration of all $2^{13}$ tested sets, implemented by the coloring/backtracking procedure described in Appendix~\ref{app:computational-details}, gives
\begin{equation}
\min\{|T|:\chi(H_T)>3\}=8.
\end{equation}
There are exactly three minimizers:
\begin{align}
&\{1,2,4,5,6,7,8,9\},\nonumber\\
&\{1,3,4,5,6,7,8,9\},\nonumber\\
&\{2,3,4,5,6,7,8,9\}.
\label{eq:minimum-tested-sets}
\end{align}
The same verification shows that deleting any one vertex of $G_{13}$ produces a $3$-colorable graph. Since $\chi(G_{13})=4$, this proves that $G_{13}$ is $4$-vertex-critical. Its displayed three-dimensional orthogonal representation supplies the separate orthogonal-rank qualification relevant here. This does not exclude a different graph with fewer preparations, or with a smaller biclique-cover number, from defining another binary-output qutrit perfect game.

\section{All-dimensional apex families}
\label{sec:apex}

The $G_{13}$ example realizes the binary structural mechanism in dimension three. Apex extension shows that the same mechanism generates perfect same-dimensional games systematically in every dimension. The one-apex graph was previously highlighted in the study of quantum colorings \cite{MancinskaRoberson2016}; here we iterate the operation and derive exact symmetric and Bob-input-compressed PM values.

For $t\ge0$, let
\begin{equation}
G_t:=G_{13}\vee K_t
\end{equation}
be the graph obtained by adjoining $t$ mutually adjacent apex vertices, each adjacent to every vertex of $G_{13}$; for $t=0$, $K_0$ denotes the empty graph. We write
\begin{equation}
\begin{aligned}
d_t&:=3+t,\\
n_t&:=|V(G_t)|=13+t,\\
m_t&:=|E(G_t)|=24+13t+\frac{t(t-1)}{2}.
\end{aligned}
\label{eq:apex-nm}
\end{equation}
Because graph join adds chromatic numbers,
\begin{equation}
\chi(G_t)=4+t.
\end{equation}
Embedding the $G_{13}$ vectors in the first three coordinates of $\C^{d_t}$ and assigning a fresh basis vector to each apex gives $\xi_{\C}(G_t)\le d_t$. Conversely, the triangle $\{1,2,3\}$ together with the apex clique forms $K_{d_t}$, so
\begin{equation}
\xi_{\C}(G_t)=d_t.
\label{eq:apex-orthogonal-rank}
\end{equation}
Hence $G_t$ gives a perfect same-dimensional separation in dimension $d_t$ for every $t\ge0$.

\subsection{Symmetric family}

\begin{theorem}[Exact symmetric apex values]
\label{thm:apex-values}
For every $t\ge0$,
\begin{align}
S_t&=Q_{d_t}(G_t)=t^2+26t+61,
\label{eq:apex-S}\\
C_{d_t}(G_t)&=t^2+26t+59=S_t-2.
\label{eq:apex-C}
\end{align}
\end{theorem}

\begin{proof}
The first equality follows from $S_t=n_t+2m_t$ and the orthogonal representation. It remains to maximize the number of locally proper vertices over color assignments $f:V(G_t)\to[d_t]$. Let $A_t$ be the apex clique and let $r$ be the number of distinct colors used on $A_t$.

If $r<t$, some apex color is repeated. At most $r-1$ apex vertices can then be locally proper, while at most all $13$ base vertices are locally proper. Hence
\begin{equation}
\ell_{V(G_t)}(f)\le13+r-1\le t+11.
\end{equation}

Suppose $r=t$, so all apex colors are distinct, and let $s$ be the number of these colors that also occur on the base graph. At most $t-s$ apex vertices are locally proper. If $s=0$, the base uses at most three colors, so non-$3$-colorability forces a monochromatic edge and at most $11$ locally proper base vertices. If $s=1$, at least one base vertex shares a color with an adjacent apex, leaving at most $12$ locally proper base vertices. If $s\ge2$, the trivial base bound is $13$. Therefore
\begin{equation}
\ell_{V(G_t)}(f)
\le
\begin{cases}
t+11,&s=0,\\
(t-1)+12=t+11,&s=1,\\
(t-s)+13\le t+11,&s\ge2.
\end{cases}
\end{equation}
The upper bound is attained by using the near-coloring \eqref{eq:near-coloring} on the base and $t$ fresh colors on the apex vertices. Thus $\delta_{d_t}(G_t,V(G_t))=2$, and Proposition~\ref{prop:graph-classical-value} gives Eq.~\eqref{eq:apex-C}.
\end{proof}

\subsection{Compressed family}

Let $A_t$ be the apex set and test
\begin{equation}
T_t:=T_8\cup A_t.
\end{equation}
Since $T_8$ covers every base edge and each apex is tested, $T_t$ is a vertex cover of $G_t$. The eight tested base vertices have degree $4+t$, while each apex has degree $12+t$. Therefore
\begin{equation}
S_t^{\mathrm{comp}} =|T_t|+\sum_{y\in T_t}\deg(y) =t^2+21t+40.
\label{eq:compressed-apex-S}
\end{equation}

\begin{theorem}[Compressed apex values]
\label{thm:compressed-apex}
For every $t\ge0$, the vertex-test game $(G_t,T_t)$ has
\begin{equation}
(X,Y,B)=(13+t,8+t,2)
\end{equation}
and
\begin{equation}
C_{d_t}(G_t,T_t)=S_t^{\mathrm{comp}}-1, \qquad Q_{d_t}(G_t,T_t)=S_t^{\mathrm{comp}}.
\label{eq:compressed-apex-values}
\end{equation}
\end{theorem}

\begin{proof}
The orthogonal representation gives the quantum value. Let $f:V(G_t)\to[d_t]$ be any color assignment. Since $T_t$ is a vertex cover and $\chi(G_t)=4+t$, at least one tested vertex is conflicting under $f$, so the classical value is at most $S_t^{\mathrm{comp}}-1$. For achievability, use the near-coloring on the base and assign a fresh color to every apex. The unique monochromatic edge is $\{1,2\}$, of which only vertex $2$ is tested, while every apex is locally proper. Hence precisely one tested vertex is conflicting under this assignment, and the upper bound is attained.
\end{proof}

\subsection{Correlator form and robustness}
\label{subsec:robustness}

We evaluate three noise models for the canonical graph strategy of Theorem~\ref{thm:graph-cover}, using the displayed rank-one preparations and projective measurements. In the preparation-noise calculation the ideal measurements are kept fixed; the resulting threshold is therefore not claimed to be globally optimized over all measurements for the noisy ensemble.

For binary outputs define
\begin{equation}
E_{xy}:=p(0|x,y)-p(1|x,y).
\end{equation}
For the symmetric graph functional,
\begin{equation}
I_G=\frac{S_G+J_G}{2}, \qquad J_G= \sum_{\{x,y\}\in E}\bigl(E_{xy}+E_{yx}\bigr) - \sum_{x\in V}E_{xx}.
\label{eq:Jgraph}
\end{equation}
For the symmetric apex family, set $J_t:=J_{G_t}$. Its ideal quantum value is $J_t=S_t$, while Theorem~\ref{thm:apex-values} gives the classical inequality
\begin{equation}
J_t\le \beta_t, \qquad \beta_t:=2C_{d_t}(G_t)-S_t=S_t-4.
\label{eq:apex-correlator-bound}
\end{equation}

Let $p_Q$ denote the ideal perfect quantum behavior. Under uniform output depolarization,
\begin{equation}
p_v(b|x,y)=v p_Q(b|x,y)+(1-v)\frac{1}{2},
\end{equation}
all correlators are multiplied by $v$, and therefore
\begin{equation}
v_{\mathrm{crit}}^{\mathrm{out}}(t) = 1-\frac{4}{t^2+26t+61}.
\label{eq:vcrit-output}
\end{equation}

For preparation white noise, keep the canonical projective measurements fixed and replace the preparations by
\begin{equation}
\rho_x(v)=v\rho_x+(1-v)\frac{\I_{d_t}}{d_t}.
\end{equation}
Because $M_{1|y}=\ket{\psi_y}\!\bra{\psi_y}$ has unit trace, the maximally mixed contribution has $E_{xy}=(d_t-2)/d_t$. Let $J_t^{\mathrm{mix}}$ denote the value of $J_t$ when every preparation is replaced by $\I_{d_t}/d_t$. Then
\begin{equation}
S_t-J_t^{\mathrm{mix}} = \frac{2}{d_t}\bigl[(d_t-1)n_t+2m_t\bigr].
\end{equation}
The fixed-measurement violation threshold is
\begin{equation}
v_{\mathrm{crit}}^{\mathrm{state}}(t) = 1- \frac{2d_t}{(d_t-1)n_t+2m_t} = 1-\frac{t+3}{t^2+20t+37}.
\label{eq:vcrit-state}
\end{equation}

Finally, consider uniform Bob-side detection efficiency $\eta$, followed by the deterministic assignment of every no-click event to output $0$, corresponding to the $+1$ outcome convention of Ref.~\cite{Divianszky2023}. This is not a general adversarial detection-loophole model. For every measurement block the sum of witness coefficients is $\deg(y)-1>0$, so output $0$ is the optimal deterministic no-click assignment within this model. The lossy value is
\begin{equation}
J_t(\eta) = \eta S_t+(1-\eta)(2m_t-n_t),
\end{equation}
and the threshold is
\begin{equation}
\eta_{\mathrm{crit}}(t) = 1-\frac{2}{n_t} = \frac{11+t}{13+t}.
\label{eq:eta-crit}
\end{equation}
For $G_{13}$ the three canonical-strategy thresholds are $57/61$, $34/37$, and $11/13$, while for the one-apex graph $G_1$, which has $14$ vertices, they are $21/22$, $27/29$, and $6/7$. All tend to one along the family. Apex extension preserves exact perfection but does not improve these robustness figures.

\section{The qutrit Torpedo game}
\label{sec:torpedo}

Binary-output support games are completely governed by conflict graphs. We now turn to a genuinely nonbinary construction, where the obstruction is affine-geometric rather than graph-theoretic, and derive exact message-cost and loss bounds for restrictions of the qutrit Torpedo relation. The qutrit SIC--MUB forbidden-output incidence pattern underlying the Torpedo relation also appears, up to relabeling, in the nonlocality, steering, and tomography inequality of Ref.~\cite{Huang2021}, and was subsequently formulated as a PM information-retrieval game in Ref.~\cite{Emeriau2022}; we use the latter work's prime-power finite-field formulation. Let $d$ be an odd prime power and let $\F_d$ denote the finite field of order $d$. Alice's input is $(x,z)\in\F_d^2$, Bob's input is $q\in\F_d\cup\{\infty\}$, and the output is denoted by $b\in\F_d$ throughout this section. The unique forbidden value is
\begin{equation}
b=x\quad(q=\infty), \qquad b=qx-z\quad(q\in\F_d).
\label{eq:torpedo-forbidden}
\end{equation}
For every odd prime-power $d$, the quantum construction of Ref.~\cite{Emeriau2022} is perfect. At $d=3$, the full game has
\begin{equation}
(X,Y,B)=(9,4,3), \qquad Q_3=S=36.
\label{eq:full-torpedo-quantum}
\end{equation}

The qutrit case admits a simple affine-plane characterization. Let
\begin{equation}
h_{\infty}(x,z)=x, \qquad h_q(x,z)=qx-z \quad(q\in\F_3).
\end{equation}
The kernels of these four nonzero linear functionals are the four one-dimensional directions of $\F_3^2$.

\begin{lemma}[Torpedo triple separation]
\label{lem:torpedo-triple}
For every three distinct points of $\F_3^2$, at least one function $h_q$, with $q\in\F_3\cup\{\infty\}$, takes the three distinct values $0,1,2$.
\end{lemma}

\begin{proof}
If the three points are collinear, choose a functional whose kernel is not the direction of their affine line. Its restriction to that line is injective. If the points are noncollinear, their three pairwise differences occupy three distinct projective directions. Choose the functional whose kernel is the fourth direction. No pairwise difference lies in the kernel, so the three function values are distinct.
\end{proof}

\begin{theorem}[Exact perfect-message cost of Torpedo restrictions]
\label{thm:torpedo-message-cost}
For any restriction of the qutrit Torpedo game to $X$ distinct preparations, with all four Bob inputs $q\in\F_3\cup\{\infty\}$ retained, the minimum number of values in a classical message required for perfect winning is
\begin{equation}
M_{\mathrm{perf}}(X)=\left\lceil\frac{X}{2}\right\rceil.
\label{eq:torpedo-message-cost}
\end{equation}
\end{theorem}

\begin{proof}
A message class containing three distinct preparations is impossible. By Lemma~\ref{lem:torpedo-triple}, some Bob input $q$ makes their forbidden values equal to all three possible outputs, so Bob has no perfect response. Hence every message class has size at most two and at least $\lceil X/2\rceil$ message values are necessary. Conversely, every class of size at most two is compatible: for each $q$, at most two outputs are forbidden, and Bob chooses the remaining output. Pairing the preparations proves sufficiency.
\end{proof}

\begin{proposition}[Classical Torpedo loss bound]
\label{prop:torpedo-loss-bound}
Consider any restriction to $X$ qutrit Torpedo preparations, retain all four Bob inputs, and allow a classical $M$-level message. Let $L_{\mathrm{Tor}}$ be the number of lost promised pairs in a deterministic strategy. Then
\begin{equation}
L_{\mathrm{Tor}}\ge \max\{X-2M,0\}.
\label{eq:torpedo-general-loss}
\end{equation}
\end{proposition}

\begin{proof}
Fix one message class $\mathcal C$. For each $q$, Bob chooses one output $b_{\mathcal C,q}$. Call a preparation $(x,z)\in\mathcal C$ lossless if
\begin{equation}
h_q(x,z)\neq b_{\mathcal C,q} \qquad \text{for every }q\in\F_3\cup\{\infty\}.
\end{equation}
There cannot be three lossless preparations in $\mathcal C$: by Lemma~\ref{lem:torpedo-triple}, for some $q$ their three forbidden values are $0,1,2$, one of which must equal $b_{\mathcal C,q}$. Thus at most two preparations in each message class are lossless, and a class of size $|\mathcal C|$ contributes at least $\max\{|\mathcal C|-2,0\}$ losses. Summing over the $M$ classes gives
\begin{equation}
L_{\mathrm{Tor}} \ge \sum_{m=1}^{M}\max\{|\mathcal C_m|-2,0\} \ge X-2M,
\end{equation}
and nonnegativity gives Eq.~\eqref{eq:torpedo-general-loss}.
\end{proof}

\begin{corollary}[Full qutrit Torpedo value]
\label{cor:full-torpedo}
The full qutrit Torpedo game satisfies
\begin{equation}
C_3=33, \qquad Q_3=S=36.
\label{eq:full-torpedo}
\end{equation}
\end{corollary}

\begin{proof}
Proposition~\ref{prop:torpedo-loss-bound} with $X=9$ and $M=3$ gives $L_{\mathrm{Tor}}\ge3$. Equality is attained by the partition into the three noncollinear triples
\begin{align}
\mathcal C_1&=\{(0,0),(0,1),(1,0)\},\nonumber\\
\mathcal C_2&=\{(0,2),(1,1),(2,1)\},\nonumber\\
\mathcal C_3&=\{(1,2),(2,0),(2,2)\}.
\end{align}
For a noncollinear triple, precisely one of the four functionals $h_q$ is injective. That Bob input forces one loss, while for each other input at most two outputs are forbidden and Bob avoids any loss. Hence the partition loses exactly three of the $36$ promised pairs.
\end{proof}

\begin{corollary}[Minimal qutrit Torpedo separation]
\label{cor:seven-torpedo}
Every restriction to seven distinct qutrit Torpedo preparations, retaining all four Bob inputs, satisfies
\begin{equation}
C_3=27, \qquad Q_3=S=28,
\end{equation}
and seven preparations are minimal for a perfect qutrit-over-trit separation within this construction.
\end{corollary}

\begin{proof}
Six preparations can be partitioned into three pairs and are therefore perfectly encoded by a trit. For seven preparations, Proposition~\ref{prop:torpedo-loss-bound} gives at least one loss.

Any seven-point subset of $\F_3^2$ contains a noncollinear triple. Encode that triple in one message and partition the four remaining preparations into two pairs. For a noncollinear triple, precisely one of the four functionals is injective, causing exactly one unavoidable loss, while the pairs cause none. Thus the score $27$ is attained.
\end{proof}

For example, one may use the seven preparations
\begin{equation}
(0,0),(0,1),(0,2),(1,0),(1,1),(1,2),(2,0).
\end{equation}
The result above shows that this restriction is not exceptional: every seven-preparation restriction retaining the four Bob inputs has the same exact trit value.

\section{Antidistinguishability relation games}
\label{sec:AD}

As a second illustration of the higher-output regime, we consider exclusion tasks generated by antidistinguishable triples. This example complements the affine Torpedo construction and yields a closed-form classical bound for arbitrary message size.

Let $\{\ket{\psi_i}\}_{i=1}^{X}$ be a set of states such that every triple is antidistinguishable. For each three-element subset $R\in\binom{[X]}{3}$, fix an ordering $R=(r_1,r_2,r_3)$ and a three-outcome POVM $\{M_{b|R}\}_{b=1}^{3}$ satisfying
\begin{equation}
\bra{\psi_{r_b}}M_{b|R}\ket{\psi_{r_b}}=0 \qquad (b=1,2,3).
\end{equation}
Alice receives $i$, Bob receives a triple $R$ containing $i$, and Bob outputs $b\in\{1,2,3\}$. The winning condition is $r_b\neq i$. The promised score is
\begin{equation}
I_{\mathrm{AD}} = \sum_{R\in\binom{[X]}{3}} \sum_{i\in R} \sum_{b=1}^{3} \ind{i\neq r_b}p(b|i,R),
\label{eq:AD-score}
\end{equation}
where $\ind{\cdot}$ denotes the indicator function. The algebraic value is
\begin{equation}
S_{\mathrm{AD}}=3\binom{X}{3}.
\end{equation}
Antidistinguishability supplies a perfect quantum strategy.

\begin{proposition}[Exact classical antidistinguishability bound]
\label{prop:AD-bound}
For a classical $M$-level message, let $C_M^{\mathrm{AD}}(X)$ denote the corresponding classical value and write $X=Mq+r$ with $0\le r<M$. Then
\begin{equation}
C_M^{\mathrm{AD}}(X) = 3\binom{X}{3} - \left[ (M-r)\binom{q}{3} +r\binom{q+1}{3} \right].
\label{eq:AD-classical-bound}
\end{equation}
\end{proposition}

The proof is given in Appendix~\ref{app:AD-bound}. A qutrit SIC has $X=9$, and every triple is antidistinguishable \cite{Renes2004,Havlicek2020}. Hence
\begin{equation}
\begin{aligned}
(X,Y,B)&=(9,84,3),\\
C_3^{\mathrm{AD}}(9)&=249,
\qquad
Q_3=S=252.
\end{aligned}
\label{eq:SIC-AD-values}
\end{equation}
The promise $i\in R$ removes input pairs on which every output is automatically winning and is equivalent to the original relation at the level of unavoidable losses \cite{Havlicek2020}. Random exclusion codes impose the additional requirement that successful outcomes occur uniformly and without error \cite{Bae2026}; because that condition constrains nonzero probabilities, it lies beyond support constraints alone.

The explicit perfect same-dimensional games discussed above are summarized in Table~\ref{tab:summary}. Explicit forms of the principal dimension-bounded inequalities are collected in Appendix~\ref{app:explicit-inequalities}.

\begin{table}[b]
\caption{Perfect same-dimensional PM games discussed in the text. Here $S$ is the algebraic value and $d_t=3+t$ in the apex rows.}
\label{tab:summary}
\centering
\scriptsize
\begin{tabular}{@{}lcc@{}}
\toprule
Game & $(X,Y,B)$ & Values\\
\midrule
$G_{13}$ symmetric & $(13,13,2)$ & $C_3=59$, $Q_3=61$\\
$G_{13}$ compressed & $(13,8,2)$ & $C_3=39$, $Q_3=40$\\
Apex symmetric & $(13+t,13+t,2)$ & $C_{d_t}=S-2$, $Q_{d_t}=S$\\
Apex compressed & $(13+t,8+t,2)$ & $C_{d_t}=S-1$, $Q_{d_t}=S$\\
Torpedo & $(9,4,3)$ & $C_3=33$, $Q_3=36$\\
Seven-point Torpedo & $(7,4,3)$ & $C_3=27$, $Q_3=28$\\
Seven-point three-input & $(7,3,3)$ & $C_3=20$, $Q_3=21$\\
Two-input dimension five & $(20,2,6)$ & $C_5=39$, $Q_5=40$\\
SIC antidist. & $(9,84,3)$ & $C_3^{\mathrm{AD}}=249$, $Q_3=252$\\
\bottomrule
\end{tabular}
\end{table}

\section{Connections and discussion}

The manuscript is organized around three structural conclusions. First, Theorem~\ref{thm:binary-characterization} gives a complete classification of binary-output exact-support advantages: quantum perfection is equivalent to a low-dimensional orthogonal representation of the conflict graph, classical perfection is equivalent to graph coloring, and Proposition~\ref{prop:biclique-cover} identifies the minimum number of Bob inputs with the edge biclique-cover number. The chromatic-number--orthogonal-rank characterization for graph-defined promise equality was known previously \cite{deWolf2001,Briet2015}, and Ref.~\cite{Prakash2026} gives a recent graph-defined realization in quantum finite automata. The new step here is the universal reduction of arbitrary binary-output PM support games to that structure, together with exact operational compression.

Second, Theorem~\ref{thm:no-qubit} is independent of output cardinality: every support relation realized perfectly by a qubit admits a perfect deterministic classical-bit realization. This sharply separates support constraints from full-statistics simulation. Two classical bits are sufficient and necessary to reproduce arbitrary qubit PM statistics in the unrestricted setting \cite{Renner2023}; restricted scenarios can already witness this cost with six preparations and five measurements, while real-qutrit behaviors can require at least five classical message values \cite{Schlosser2026}. Exact probabilities and exact supports therefore define genuinely different communication resources.

Third, Theorem~\ref{thm:no-qutrit-two-input} shows that two Bob inputs remain insufficient for a qutrit even with arbitrary finite output alphabets and POVMs, while Proposition~\ref{prop:seven-three-three} shows that three Bob inputs suffice. Thus $Y=3$ is the exact qutrit threshold. In contrast, Theorem~\ref{thm:d5-two-input} gives a genuine $Y=2$ separation in dimension five. The only unresolved dimension between these statements is four, giving the window in Eq.~\eqref{eq:y2-dimension-window}.

The constructions illustrate the scope and boundaries of these three conclusions. The compressed $G_{13}$ game is the flagship finite realization of the binary theorem and minimizes Bob's input alphabet for its conflict graph. The apex family propagates the same graph mechanism to every dimension. The Torpedo and antidistinguishability games instead probe the higher-output regime, where bounded-rank conflict hypergraphs replace ordinary graphs and additional affine or exclusion structure becomes relevant. A conflict graph nevertheless specifies only the obstruction, not a unique task: Appendix~\ref{app:edge-game} gives an alternative edge-identification realization of $G_{13}$ with $C_3=47<Q_3=S=48$.

The graph formulation also clarifies the relation with contextuality. Ramanathan and Horodecki associated state-independent contextuality with the obstruction $\chi_f(G)>d$, where $\chi_f(G)$ denotes the fractional chromatic number \cite{RamanathanHorodecki2014}; later work refined the realization-dependent conditions \cite{Cabello2015}. Since $\chi_f(G)\le\chi(G)$, state-independent contextuality implies the non-$d$-colorability used here, but not conversely. A graph-based perfect PM game need not constitute a state-independent contextuality proof. In bipartite nonlocal games, perfect strategies are more restrictive and are connected to Kochen--Specker sets \cite{Cabello2025}.

The examples also expose a distinction between combinatorial size and robustness. The compressed $G_{13}$ game is input-minimal for its conflict graph, while apex extension preserves perfection in all dimensions; nevertheless, the relative classical deficits $2/S_t$ and $1/S_t^{\mathrm{comp}}$, together with the canonical-strategy noise and loss tolerances $1-v_{\mathrm{crit}}^{\mathrm{out}}$, $1-v_{\mathrm{crit}}^{\mathrm{state}}$, and $1-\eta_{\mathrm{crit}}$, vanish asymptotically. Equivalently, the corresponding critical thresholds tend to one. A natural optimization problem is therefore to find graphs or bounded-rank hypergraphs for which $(S-C_d)/S$ remains bounded away from zero as the game size increases.

Several basic questions remain open. Most immediately, does a four-dimensional support game with two Bob inputs exhibit a perfect same-dimensional separation, thereby making $d_{\min}^{(Y=2)}=4$, or is the dimension-five construction optimal? Is $G_{13}$ of minimum order among graphs satisfying $\xi_{\C}(G)\le3<\chi(G)$? More generally, what are the minimum possible values of $|V(G)|$ and $\operatorname{bc}(G)$ among such graphs? Can the extremal cases be characterized directly through critical graphs of orthogonal rank three? Which bounded-rank conflict hypergraphs are realizable by quantum supports in fixed dimension? Can exclusion-based games be compressed as efficiently as graph games? Finally, experimentally useful perfect games should optimize not only alphabet sizes but also visibility, detection efficiency, and tolerance to imperfect dimension control.

\begin{acknowledgments}
The author thanks Ad\'an Cabello, Armin Tavakoli, Ricardo Faleiro and Flavien Hirsch for useful discussions and comments. This work was funded by national funds through FCT---Funda\c{c}\~ao para a Ci\^encia e a Tecnologia, I.P., and, when eligible, co-funded by European Union funds under project UID/50008/2025---Instituto de Telecomunica\c{c}\~oes, DOI: \href{https://doi.org/10.54499/UID/50008/2025}{10.54499/UID/50008/2025}. OpenAI ChatGPT (GPT-5.4, GPT-5.5, and GPT-5.6) was used to assist with code development, exact algebraic and computational cross-checks, the exploration and refinement of candidate mathematical arguments, and manuscript organization and editing. The author specified the problems, algorithms, assumptions, and conventions; verified the resulting scientific claims and computational outputs; revised all AI-assisted material; and takes full responsibility for the content of the manuscript.
\end{acknowledgments}

\section*{Data and code availability}
All finite combinatorial data needed to reproduce the results are contained in the article. The complete data underlying the results reported in this work, including a standalone Python verification script that reconstructs $G_{13}$ from the ray vectors and reproduces the independent-set, coloring, biclique-cover, tested-set, vertex-deletion, and Torpedo enumerations quoted in the text, will be deposited in a public repository no later than publication of the peer-reviewed article.

\FloatBarrier

\appendix

\section{Proof of the rank-three Parseval transversal lemma}
\label{app:parseval-transversal}

We prove Lemma~\ref{lem:parseval-transversal}. Choose arbitrary triples of vectors $e_1,e_2,e_3$ from $\Phi$ and $f_1,f_2,f_3$ from $\Psi$ that form bases $E=(e_1,e_2,e_3)$ and $F=(f_1,f_2,f_3)$, and write
\begin{equation}
f_j=\sum_{i=1}^3 t_{ij}e_i,
\qquad T=(t_{ij})\in\mathrm{GL}(3,\C).
\end{equation}
For a fixed pairing, replacing $e_i$ by $f_j$ gives a basis exactly when $t_{ij}\neq0$, while replacing the complementary two vectors gives a basis exactly when $(T^{-1})_{ji}\neq0$. Introduce
\begin{equation}
p_{ij}:=t_{ij}(T^{-1})_{ji}.
\label{eq:pij}
\end{equation}
The bipartite graph with an edge $i\sim j$ whenever $p_{ij}\neq0$ therefore has a perfect matching exactly when the two bases can be ordered so that all eight transversals are bases. Moreover,
\begin{equation}
\sum_jp_{ij}=1,
\qquad
\sum_ip_{ij}=1,
\label{eq:pij-sums}
\end{equation}
by $TT^{-1}=T^{-1}T=\I$.

Suppose there is no perfect matching. Equation~\eqref{eq:pij-sums} implies that every row and column has a nonzero entry. Hall's theorem then allows relabeling so that rows $2,3$ have the unique common neighbor column $1$. Hence
\begin{equation}
p_{22}=p_{23}=p_{32}=p_{33}=0,
\end{equation}
while Eq.~\eqref{eq:pij-sums} gives
\begin{equation}
p_{21}=p_{31}=p_{12}=p_{13}=1,
\qquad p_{11}=-1.
\end{equation}
All entries in the first row and column of $T$ are therefore nonzero. Independent rescalings of basis vectors preserve all linear-dependence relations, so we can normalize
\begin{equation}
T=
\begin{pmatrix}
1&1&1\\
1&a&b\\
1&c&d
\end{pmatrix}.
\end{equation}
The four missing edges imply
\begin{equation}
\begin{aligned}
a(d-1)&=0,& b(1-c)&=0,\\
c(1-b)&=0,& d(a-1)&=0.
\end{aligned}
\end{equation}
Thus $(a,d)$ and $(b,c)$ are each either $(0,0)$ or $(1,1)$. Invertibility excludes choosing the same option twice. After interchanging labels if necessary, every bad pair is therefore equivalent to
\begin{equation}
E=\{(1,0,0),(0,1,0),(0,0,1)\},
\end{equation}
\begin{equation}
F=\{f_1,f_2,f_3\}
=\{(1,1,1),(1,1,0),(1,0,1)\}.
\label{eq:canonical-bad-pair}
\end{equation}

We next show that the obstruction in Eq.~\eqref{eq:canonical-bad-pair} cannot be extended by a new ray. Let $v=(x,y,z)\neq0$. For a basis $G=(g_1,g_2,g_3)$, the zero/nonzero pattern relevant to symmetric exchanges with the standard basis is that of
\begin{equation}
q_{ij}=g_{ij}\,C_{ij}(G),
\end{equation}
where $C_{ij}(G)$ is the corresponding cofactor. Replacing $f_2$ by $v$ gives determinant $z-x$ and
\begin{equation}
Q_2=
\begin{pmatrix}
y&-x&z-y\\
z-x&0&0\\
-y&z&y-x
\end{pmatrix}.
\label{eq:Q2}
\end{equation}
If $x\neq z$, row $2$ must match column $1$, and a perfect matching exists whenever either $x\neq0$ and $y\neq x$, or $z\neq0$ and $z\neq y$. If both alternatives fail, then either $v\parallel(1,1,0)=f_2$, or $v=(0,z,z)$ with $z\neq0$. In the latter case replace $f_1$ instead. The exchange matrix is
\begin{equation}
Q_1=
\begin{pmatrix}
x&-y&-z\\
-y&x-z&0\\
-z&0&x-y
\end{pmatrix},
\end{equation}
and at $v=(0,z,z)$ it has the perfect matching $1\mapsto2$, $2\mapsto1$, $3\mapsto3$, while the replacement determinant is $-2z\neq0$.

It remains to consider $x=z$. If also $x=y$, then $v\parallel f_1$. Otherwise replace $f_3$. Its determinant is $y-x\neq0$ and
\begin{equation}
Q_3=
\begin{pmatrix}
z&y-z&-x\\
-z&z-x&y\\
y-x&0&0
\end{pmatrix}.
\label{eq:Q3}
\end{equation}
Here row $3$ must match column $1$, and the matching $1\mapsto2$, $2\mapsto3$, $3\mapsto1$ works whenever $y\neq0$. The only remaining case has $y=0$ and hence $v\parallel(1,0,1)=f_3$. We have proved that any new ray not parallel to one of $f_1,f_2,f_3$ creates, by replacing one element of $F$, a basis admitting a fully transversal pairing with $E$.

Now suppose, toward a contradiction, that the two Parseval frames $\Phi$ and $\Psi$ contain no pair of bases satisfying Lemma~\ref{lem:parseval-transversal}. Starting from a bad pair of bases $E$ from $\Phi$ and $F$ from $\Psi$, the extension result forces every nonzero vector of $\Psi$ to be parallel to one of $f_1,f_2,f_3$. By symmetry, every nonzero vector of $\Phi$ is parallel to one of $e_1,e_2,e_3$. Parseval tightness therefore has the form
\begin{equation}
\I_3=\sum_{i=1}^3\alpha_i\ket{\hat e_i}\!\bra{\hat e_i},
\qquad \alpha_i>0,
\end{equation}
where $\hat e_i$ are unit vectors on the three rays. The square matrix with columns $\sqrt{\alpha_i}\hat e_i$ satisfies $WW^\dagger=\I_3$, hence is unitary; the three rays are therefore mutually orthogonal. The same holds for the $f_i$.

Finally, any two orthonormal bases in dimension three admit a fully transversal pairing. Let $U_{ij}=\langle e_i|f_j\rangle$ be their unitary change-of-basis matrix. Since $\det U\neq0$, some term in its determinant expansion is nonzero, so after permuting the $f_j$ all three diagonal entries are nonzero. For a unitary matrix the cofactor of a diagonal entry is $\det(U)$ times its complex conjugate, and is therefore also nonzero. Thus every one-element and every complementary two-element replacement is a basis, as are the two original bases. This contradiction proves Lemma~\ref{lem:parseval-transversal}.

\section{Explicit graph data and finite certificates}
\label{app:computational-details}

The edge set corresponding to Eq.~\eqref{eq:G13-vectors} is
\begin{equation}
\begin{aligned}
E(G_{13})=\{&
\{1,2\},\{1,3\},\{1,4\},\{1,5\},\\
& \{2,3\},\{2,6\},\{2,7\},\{3,8\},\\
& \{3,9\},\{4,5\},\{4,11\},\{4,12\},\\
& \{5,10\},\{5,13\},\{6,7\},\{6,11\},\\
& \{6,13\},\{7,10\},\{7,12\},\{8,9\},\\
& \{8,12\},\{8,13\},\{9,10\},\{9,11\}
\}.
\end{aligned}
\label{eq:G13-edges}
\end{equation}
Direct dot products verify that these and only these pairs are orthogonal.

The exact coloring routine uses DSATUR-style recursive backtracking \cite{Brelaz1979}. At each node it selects an uncolored vertex maximizing the number of distinct colors already present among its colored neighbors, breaking ties by degree, and branches over all admissible colors. Exhaustion of the search tree certifies noncolorability. Tested sets are enumerated as vertex subsets, and the same exact coloring test is applied to every tested-edge graph and every one-vertex deletion.

The verification script performs the following exact checks.
\begin{enumerate}
\item Enumeration of all $2^{13}$ vertex subsets gives
\begin{equation}
\alpha(G_{13})=5,
\end{equation}
with certificate $\{1,10,11,12,13\}$. Hence $\tau(G_{13})=8$. Direct enumeration of common-neighbor pairs finds no $4$-cycle. Consequently every biclique of $G_{13}$ is a star and
\begin{equation}
\operatorname{bc}(G_{13})=\tau(G_{13})=8,
\end{equation}
which certifies the minimum Bob-input count in Eq.~\eqref{eq:G13-biclique-cover}.
\item Exact backtracking finds no $3$-coloring and finds the $4$-coloring
\begin{equation}
(1,2,3,2,3,1,3,2,1,2,3,1,4),
\end{equation}
where entries are listed in vertex order.
\item Enumeration of the $3^{13}$ trit encodings gives
\begin{equation}
\delta_3(G_{13},V(G_{13}))=2, \qquad \delta_3(G_{13},T_8)=1.
\end{equation}
The encoding \eqref{eq:near-coloring} attains both minima.
\item Enumeration of all tested sets gives Eq.~\eqref{eq:minimum-tested-sets}. For each of the thirteen one-vertex deletions, exact backtracking returns a $3$-coloring. Explicit certificates are listed in Table~\ref{tab:deletion-colorings}.
\item The script verifies Lemma~\ref{lem:torpedo-triple} for all $\binom93=84$ triples and independently cross-checks the analytic values $C_3=33$ for the full qutrit Torpedo game and $C_3=27$ for each of its $\binom97=36$ seven-point restrictions.
\end{enumerate}

\begin{table*}[t]
\caption{Explicit $3$-colorings of every one-vertex deletion of $G_{13}$. Each certificate lists colors in vertex order, with a dash at the deleted vertex.}
\label{tab:deletion-colorings}
\centering
\scriptsize
\renewcommand{\arraystretch}{1.05}
\begin{ruledtabular}
\begin{tabular}{c@{\quad}l@{\qquad\qquad}c@{\quad}l}
Deleted & Coloring certificate & Deleted & Coloring certificate\\
$1$ & $(-,3,1,1,2,1,2,2,3,1,2,3,3)$ & $8$ & $(1,2,3,2,3,1,3,-,1,2,3,1,2)$\\
$2$ & $(3,-,1,1,2,1,2,2,3,1,2,3,3)$ & $9$ & $(1,2,3,2,3,3,1,1,-,2,1,3,2)$\\
$3$ & $(3,1,-,1,2,2,3,1,2,1,3,2,3)$ & $10$ & $(1,2,3,2,3,3,1,1,2,-,1,3,2)$\\
$4$ & $(3,1,2,-,1,2,3,1,3,2,1,2,3)$ & $11$ & $(1,2,3,3,2,1,3,1,2,1,-,2,3)$\\
$5$ & $(3,1,2,1,-,2,3,3,1,2,3,2,1)$ & $12$ & $(1,2,3,3,2,3,1,2,1,3,2,-,1)$\\
$6$ & $(1,3,2,2,3,-,1,1,3,2,1,3,2)$ & $13$ & $(1,2,3,2,3,1,3,2,1,2,3,1,-)$\\
$7$ & $(1,3,2,2,3,1,-,3,1,2,3,1,2)$ & & \\
\end{tabular}
\end{ruledtabular}
\end{table*}

\FloatBarrier

\section{Alternative edge-identification realization}
\label{app:edge-game}

The same graph obstruction admits a second canonical binary realization. Fix an orientation of every edge $e=(u_e,v_e)$. Alice receives a vertex, Bob receives an edge, and the promise is that Alice's vertex is one endpoint of Bob's edge. Bob outputs which endpoint was sent. The score is
\begin{equation}
I_G^{\mathrm{edge}} = \sum_{e\in E} \left[ p(0|u_e,e)+p(1|v_e,e) \right],
\label{eq:edge-game}
\end{equation}
with algebraic value $S_G^{\mathrm{edge}}:=2|E(G)|$. If $\xi_{\C}(G)\le d$, orthogonal endpoint states are perfectly distinguished by a binary measurement. Classically, an edge is perfect exactly when its endpoints receive distinct messages. A monochromatic edge forces precisely one loss, and therefore
\begin{equation}
C_d^{\mathrm{edge}}(G) = 2|E(G)|-\mu_d(G),
\label{eq:edge-game-value}
\end{equation}
where
\begin{equation}
\mu_d(G) := \min_{f:V\to[d]} \left| \left\{ \{u,v\}\in E:f(u)=f(v) \right\} \right|
\end{equation}
is the minimum number of monochromatic edges over all $d$-color assignments. For $G_{13}$, $\mu_3(G_{13})=1$ by Eq.~\eqref{eq:near-coloring}, giving
\begin{equation}
(X,Y,B)=(13,24,2), \qquad C_3=47, \qquad Q_3=S=48.
\end{equation}
This illustrates that a conflict graph specifies the classical obstruction but not a unique operational game.

\section{Classical bound for antidistinguishability games}
\label{app:AD-bound}

Fix a deterministic encoding $e:[X]\to[M]$, and let its message classes $\mathcal C_1,\ldots,\mathcal C_M$ have sizes $n_1,\ldots,n_M$. For a message class $\mathcal C_m$ and a Bob triple $R=(r_1,r_2,r_3)$, there is no unavoidable loss when $R\nsubseteq\mathcal C_m$: Bob chooses an output $b$ such that $r_b\in R\setminus\mathcal C_m$. Then $r_b$ differs from every possible Alice input in $\mathcal C_m\cap R$. If $R\subseteq\mathcal C_m$, then for every output $b$, the preparation $i=r_b$ is a possible Alice input, and exactly one promised pair is lost. The number of unavoidable losses is therefore
\begin{equation}
L(e)=\sum_{m=1}^{M}\binom{n_m}{3}.
\end{equation}
Consequently,
\begin{equation}
C_M^{\mathrm{AD}}(X)
=
3\binom{X}{3}
-
\min_{\substack{n_1+\cdots+n_M=X\\n_m\ge0}}
\sum_{m=1}^{M}\binom{n_m}{3}.
\end{equation}
The function $n\mapsto\binom{n}{3}$ is discretely convex. If $a\ge b+2$, then
\begin{equation}
\binom{a}{3}+\binom{b}{3} - \binom{a-1}{3}-\binom{b+1}{3} = \binom{a-1}{2}-\binom{b}{2} \ge0.
\end{equation}
Thus the minimum is attained by balancing the message-class sizes. Writing $X=Mq+r$, with $0\le r<M$, gives $M-r$ classes of size $q$ and $r$ classes of size $q+1$, proving Eq.~\eqref{eq:AD-classical-bound}.

\section{Explicit forms of the principal inequalities}
\label{app:explicit-inequalities}

This appendix collects the main dimension-bounded PM inequalities in a form intended for direct use. Probabilities that do not appear have coefficient zero. The bounds stated below are the classical bounds for a message of the same dimension as the corresponding quantum system.

\subsection{The symmetric and compressed \texorpdfstring{$G_{13}$}{G13} inequalities}

For convenience, the neighborhoods of the vertices in the labeling of Eq.~\eqref{eq:G13-vectors} are
\begin{equation}
\begin{array}{c@{:\ }l@{\qquad}c@{:\ }l}
1&\{2,3,4,5\} & 8&\{3,9,12,13\}\\
2&\{1,3,6,7\} & 9&\{3,8,10,11\}\\
3&\{1,2,8,9\} & 10&\{5,7,9\}\\
4&\{1,5,11,12\} & 11&\{4,6,9\}\\
5&\{1,4,10,13\} & 12&\{4,7,8\}\\
6&\{2,7,11,13\} & 13&\{5,6,8\}\\
7&\{2,6,10,12\} & \multicolumn{2}{c}{}
\end{array}
\label{eq:G13-neighborhoods}
\end{equation}
Define the Bob-input block
\begin{equation}
\Phi_y[p] := p(1|y,y)+\sum_{x\in N_{G_{13}}(y)}p(0|x,y).
\label{eq:G13-block}
\end{equation}
The symmetric qutrit inequality is
\begin{equation}
I_{13}^{\mathrm{sym}}[p] := \sum_{y=1}^{13}\Phi_y[p] \le 59,
\label{eq:explicit-G13-symmetric}
\end{equation}
whereas a qutrit reaches the algebraic value $61$. For the minimum tested set $T_8=\{2,3,4,5,6,7,8,9\}$, the compressed inequality is
\begin{equation}
I_{13}^{\mathrm{comp}}[p] := \sum_{y=2}^{9}\Phi_y[p] \le 39,
\label{eq:explicit-G13-compressed}
\end{equation}
whereas a qutrit reaches the algebraic value $40$.

\subsection{All-dimensional apex inequalities}

Let $A_t=\{a_1,\ldots,a_t\}$ be the apex set and recall that $d_t=3+t$. For a base Bob input $y\in\{1,\ldots,13\}$ define
\begin{equation}
\Phi_y^{(t)}[p] := p(1|y,y) + \sum_{x\in N_{G_{13}}(y)}p(0|x,y) + \sum_{j=1}^{t}p(0|a_j,y),
\label{eq:apex-base-block}
\end{equation}
and for an apex Bob input $a_j$ define
\begin{equation}
\begin{split}
\Psi_j^{(t)}[p]
:=\;&p(1|a_j,a_j)
+\sum_{x=1}^{13}p(0|x,a_j)\\
&+\sum_{\substack{k=1\\k\neq j}}^{t}p(0|a_k,a_j).
\end{split}
\label{eq:apex-apex-block}
\end{equation}
The symmetric family is
\begin{equation}
\sum_{y=1}^{13}\Phi_y^{(t)}[p] + \sum_{j=1}^{t}\Psi_j^{(t)}[p] \le t^2+26t+59,
\label{eq:explicit-apex-symmetric}
\end{equation}
while a $d_t$-dimensional quantum system reaches $t^2+26t+61$. The compressed family is
\begin{equation}
\sum_{y=2}^{9}\Phi_y^{(t)}[p] + \sum_{j=1}^{t}\Psi_j^{(t)}[p] \le t^2+21t+39,
\label{eq:explicit-apex-compressed}
\end{equation}
while a $d_t$-dimensional quantum system reaches $t^2+21t+40$.

\subsection{Full qutrit Torpedo inequality}

For the full preparation set $\F_3^2$, define
\begin{equation}
L_{\mathrm{Tor},9}[p] := \sum_{(x,z)\in\F_3^2} \sum_{q\in\{\infty,0,1,2\}} p\!\left(h_q(x,z)\middle|(x,z),q\right).
\label{eq:explicit-full-torpedo-loss}
\end{equation}
Every classical trit satisfies
\begin{equation}
L_{\mathrm{Tor},9}[p]\ge3,
\label{eq:explicit-full-torpedo-bound}
\end{equation}
whereas the perfect qutrit strategy has $L_{\mathrm{Tor},9}=0$. Equivalently,
\begin{equation}
I_{\mathrm{Tor},9}[p] :=36-L_{\mathrm{Tor},9}[p] \le33,
\label{eq:explicit-full-torpedo-score}
\end{equation}
with qutrit value $36$.

\subsection{Seven-preparation restriction}

Take the preparation set
\begin{equation}
\mathcal X_7 = \{(0,0),(0,1),(0,2),(1,0),(1,1),(1,2),(2,0)\}.
\end{equation}
For each preparation, the forbidden outputs $h_q(x,z)$ are
\begin{equation}
\begin{array}{c|cccc}
(x,z) & q=\infty & q=0 & q=1 & q=2\\
\hline
(0,0) & 0 & 0 & 0 & 0\\
(0,1) & 0 & 2 & 2 & 2\\
(0,2) & 0 & 1 & 1 & 1\\
(1,0) & 1 & 0 & 1 & 2\\
(1,1) & 1 & 2 & 0 & 1\\
(1,2) & 1 & 1 & 2 & 0\\
(2,0) & 2 & 0 & 2 & 1
\end{array}
\label{eq:torpedo-forbidden-seven}
\end{equation}
where all entries are evaluated in $\F_3$.
In loss form, every classical trit satisfies
\begin{equation}
L_{\mathrm{Tor},7}[p] := \sum_{(x,z)\in\mathcal X_7} \sum_{q\in\{\infty,0,1,2\}} p\!\left(h_q(x,z)\middle|(x,z),q\right) \ge 1,
\label{eq:explicit-torpedo-loss}
\end{equation}
whereas the qutrit strategy has $L_{\mathrm{Tor},7}=0$. Equivalently, the winning-score inequality is
\begin{equation}
I_{\mathrm{Tor},7}[p] :=28-L_{\mathrm{Tor},7}[p] \le27,
\label{eq:explicit-torpedo-score}
\end{equation}
with qutrit value $28$.

\subsection{Qutrit SIC antidistinguishability inequality}

Label the nine preparations by $[9]$. Bob's inputs are the $84$ triples $R\in\binom{[9]}{3}$, each written in increasing order as $R=(r_1,r_2,r_3)$. The classical-trit inequality has the compact loss form
\begin{equation}
L_{\mathrm{SIC}}[p] := \sum_{R\in\binom{[9]}{3}} \sum_{b=1}^{3}p(b|r_b,R) \ge 3.
\label{eq:explicit-SIC-loss}
\end{equation}
Every triple of qutrit SIC states is antidistinguishable, so the qutrit value is $L_{\mathrm{SIC}}=0$. Equivalently,
\begin{equation}
I_{\mathrm{SIC}}[p] :=252-L_{\mathrm{SIC}}[p] \le249,
\label{eq:explicit-SIC-score}
\end{equation}
with qutrit value $252$.

\end{document}